\documentclass[acmtog, nonacm]{acmart}
\AtBeginDocument{%
  }

\usepackage{amsthm}
\usepackage{graphicx}
\usepackage{enumitem}
\newtheorem{theorem}{Theorem}

\usepackage[final]{pdfpages}

\begin{document}
\author{Dong Xiao}
\email{xiaodong@ustc.edu.cn}
    \affiliation{%
  \department{School of Mathematical Sciences}
  \institution{University of Science and Technology of China}
  \country{China}
}
\author{Renjie Chen}
\authornote{Corresponding author}
\email{renjiec@ustc.edu.cn}
\affiliation{%
  \department{School of Mathematical Sciences}
  \institution{University of Science and Technology of China}
 \country{China}
}
\author{Bailin Deng}
\email{dengb3@cardiff.ac.uk}
\affiliation{%
  \department{School of Computational and Mathematical Sciences}
  \institution{Cardiff University}
  \country{UK}
}





\title{Domain-Varying 2D Green's Functions for Cage-based Deformation}

\begin{abstract}
  In this work, we propose a novel theoretical view of cage-based deformation based on domain-varying Green’s functions and treat this domain as a new control space for the deformation effects. Harmonic Coordinates (HC)~\citep{2007Harmonic} and Green Coordinates (GC)~\citep{2008Green} are classic methods in cage-based deformation and serve as the theoretical foundation for shape editing in a range of practical deformation tools. Our method revisits these two classical approaches. Specifically, we propose a framework based on Green’s functions across diverse domains (independent of the cage-enclosed domain) to unify these two techniques. To our knowledge, this represents the first such attempt in nearly two decades. Based on this perspective, we propose a novel cage-based deformation technique that introduces a new control space and utilizes domain-varying Green’s functions to yield varying deformation effects. Our method also establishes a continuous transition of effects from HC to GC as the Green’s function domain $\Theta$ expands from the cage region $\Omega$ to the entire $\mathbb{R}^2$.  We call our method Domain-Varying Green Coordinates (DVGC). When $\Theta$ is a disk or a rectangle, the Green’s function possesses analytic or semi-analytic expressions, respectively, enabling the DVGC to be computed without finite element discretization. Furthermore, when $\Theta$ is a disk, the DVGC admit a closed-form expression for 2D simplicial cages, thereby eliminating the need for numerical integration. Experiments demonstrate that our method provides a novel control space ranging from `more consistent with the cage’ to `more shape-preserving’, generating diverse deformation effects by varying the Green’s function domains. The source code is available at~\url{https://github.com/Submanifold/DVGC}.
\end{abstract}

\ccsdesc[500]{Computing methodologies~Shape modeling}
\ccsdesc[500]{Computing methodologies~Geometric deformation}

\keywords{Cage-based deformation, Green's function, Green's third identity, Harmonic Coordinates}

\begin{teaserfigure}
    \centering
    \includegraphics[width=\textwidth]{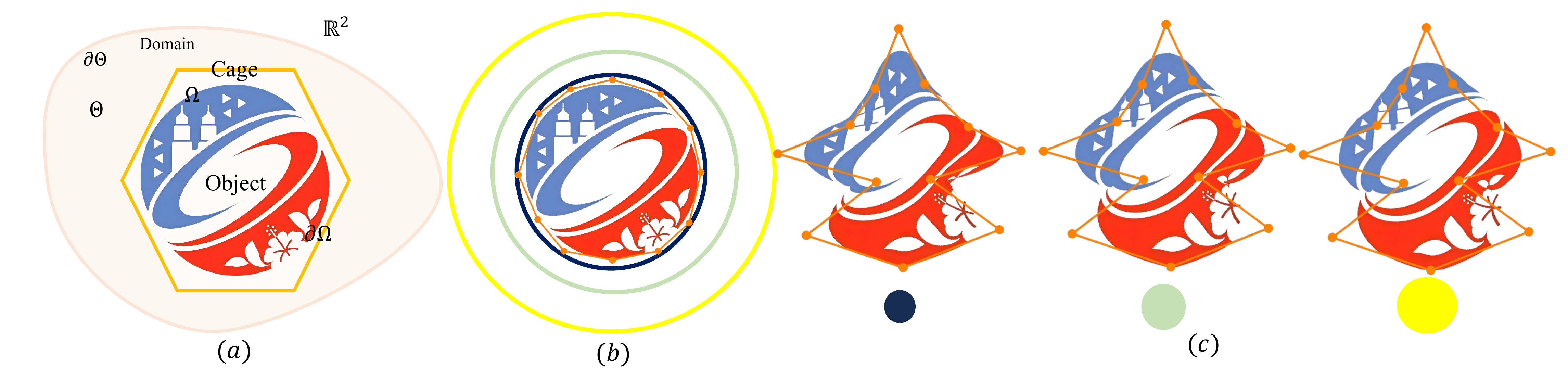}
    \caption{Illustration of our method. (a) The positional relationship among the object, the cage domain $\Omega$, and the Green's function domain $\Theta$. (b) The same source cage (orange polygonal chains) with different Green's function domains (blue, green, and yellow disks of varying radii; only their boundaries are shown). (c) Different deformation results obtained using the same target cage with varying Green's function domains.}
    \label{fig:teaser}
\end{teaserfigure}

\maketitle

\section{Introduction}
Flexible and user-friendly shape editing is a critical problem in computer graphics, especially in character animation. A representative technique is cage-based deformation, where users manipulate a coarse cage to achieve the shape design of the embedded object. The core step involves establishing the deformation coordinates, which express the position of a point inside the cage as a linear combination of source cage vertices or both vertices and normals, forming a partition of unity. Subsequently, the deformed object can be represented using these coordinates together with the target cage. Many methods exist for constructing the deformation coordinates ~\cite{2017HormannBarycentric, 2024Survey}, with each method exhibiting distinct effects and properties. Among these techniques, Mean Value Coordinates (MVC)~\citep{2005MVC} and Harmonic Coordinates (HC)~\citep{2006GeneralConstructBC, 2007Harmonic} are established based on the cage vertices, with deformation results that are more consistent with the target cage. Green Coordinates (GC)~\cite{2008Green} incorporate cage normals into the deformation coordinates, yielding shape-preserving results, which manifest as angle preservation in planar deformations. 

Both HC and GC are grounded in harmonic function theory. Specifically, HC are established by solving the Dirichlet boundary-value problem of the Laplace equation, whereas GC utilizes Green’s third identity with boundary integrals to achieve linear reproduction for interior points. Consequently, HC lack normal control, resulting in obvious shearing artifacts and faceted creasing, but yield deformation effects that adhere more strictly to the target cage. Given that cages are often coarse structures, using simplicial cages may cause the deformation results of HC to become excessively sharp or segmented. Consequently, shape preservation is an important property to consider. GC possesses normal control and maintains angle preservation with respect to the source object for 2D deformation; however, the deformation results of GC sometimes deviate noticeably from the target cage, as the angle-preserving condition is too restrictive. Despite the distinct deformation behaviors of HC and GC, we propose, to the best of our knowledge, the first unified framework to reconcile these approaches based on the generalized Green’s identity, utilizing various Green function domains $\Theta$ that encompass the cage boundary domain $\Omega$. We demonstrate that HC and GC correspond to the specific cases of $\Theta=\Omega$ and $\Theta=\mathbb{R}^2$, respectively. 

Based on the preceding observations, we propose a novel control space that leverages varying domains $\Theta$ to achieve diverse deformation effects ranging from `more consistent with the target cage’ to `more locally shape-preserving with the source object’ between HC and GC. Fig.~\ref{fig:teaser} illustrates the positional relationships among the object, the cage domain, and the Green’s function domain, as well as the core contribution of our method. Technically, our method establishes a \textbf{domain-generalized Green's third identity}, where the Green’s function is not limited to the fundamental solution of the Laplace equation, but is rather the Green’s function on any domain $\Theta$ satisfying $\Omega \subseteq \Theta \subseteq \mathbb{R}^2$. We utilize this formulation to establish deformation coordinates and designate this method as the \textbf{Domain-Varying Green Coordinates (DVGC)}; we further demonstrate that GC and HC are both special instances of DVGC.

Moreover, when $\Theta$ is a disk or a rectangle, its Green's function admits analytic or semi-analytic expressions, respectively. Therefore, we select gradually expanding regions, ranging from the smallest region enclosing $\Omega$ to the entire $\mathbb{R}^2$, thereby achieving a transition effects from HC to GC using the same set of source and target cages. Furthermore, when $\Theta$ is a disk, we can also derive the closed-form expression of DVGC, facilitating fast and accurate deformation and eliminating the need for numerical integration. We conduct extensive experiments on a wide variety of 2D cages; for each example, we demonstrate the results of the same source and target cages with different $\Theta$. Experimental results show that our method yields diverse deformation effects, providing enhanced flexibility to achieve the desired outcomes, whether prioritizing cage adherence or shape preservation. 
\section{Related Work}
Cage-based deformation is a well-established paradigm for interactive shape editing, widely utilized in both animation and digital character modeling for its intuitive control ability. A recent survey~\citep{2024Survey} provides a comprehensive summary of the main techniques of the last few decades. The choice of cage geometry (e.g., simplicial cages, quadrilateral, Bézier surfaces) and coordinate family (e.g., Generalized Barycentric Coordinates, coordinates with normal control) collectively determine the expressiveness and visual behavior of the deformation. 

\subsection{Generalized Barycentric Coordinates}
Generalized Barycentric Coordinates (GBC) \citep{2015ReviewBarycentric, 2017HormannBarycentric, 2019GeneralBarycentric, 2023VBC} primarily establish a partition-of-unity representation for positions interior to the cage and reproduce affine functions based on the cage vertices, thereby providing intuitive control over the embedded object. Closed-form Mean-Value Coordinates (MVC) are a representative type of GBC. The 2D MVC is first introduced by \citet{2003MVCPoly}, initially restricted to star-shaped polygons. The extension to arbitrary polygons is later established by \citet{2006MVCArbitrary}. In 3D, MVC are derived via the spherical projection of the cage structure \citep{2005MVCPoly, 2005MVC}. While this approach ensures computational simplicity, it suffers from the drawback of generating negative weights for concave cage regions. Subsequent work~\citep{2007PMVC} has addressed this issue by sacrificing a certain degree of smoothness. Poisson coordinates~\citep{2013Poisson} generalize MVC via the Poisson integral formula, and provide a new control parameter for MVC based on the center of the Poisson kernel. Harmonic Coordinates (HC) ensure positive weights for non-convex regions by constructing the coordinates via the boundary Dirichlet problem of the Laplace equation. This formulation first appears in \citet{2006GeneralConstructBC} and is further rigorously analyzed by \citet{2007Harmonic} in terms of mathematical properties, and is generalized for broader applications. However, HC do not possess a closed-form solution and are usually solved using the finite element method. Maximum Entropy Coordinates \citep{2008MaxEntropy} and Maximum Likelihood Coordinates \citep{2023MaximumLikelihood} are established based on entropy theory; these coordinates are non-negative and smooth, though they usually lack closed-form expressions. \citet{2024Stochastic} provide a unified framework for different cage representations and formulate the deformation task as a weighted least-squares minimization problem, solving it via the Monte Carlo integration.

\subsection{Deformation coordinates with normal control}
GBCs mainly produce affine transformations based solely on the cage vertices, which may produce shearing artifacts of the deformed objects. Furthermore, when using simplicial cages, overly sharp or segmented shapes may also appear. In contrast, normal-controlled coordinates~\citep{2008Green, 2009DriveGreen, 2012Biharmonic, 2023Somigliana} excel at preserving the local shape of the source object and preventing large shearing artifacts by incorporating boundary normal information into the deformation formulation. A typical method is Green Coordinates (GC)~\cite{2008Green, 2022QuadGreen, 2023PolyGreen, 2025BezierCage, 2025Polynomial3D}, which leverages Green’s third identity to produce a boundary integral formulation with both vertex and normal contributions. This formulation naturally ensures the linear reproduction and the partition of unity of the Dirichlet term. In the complex plane, Cauchy Coordinates~\cite{2009CBC} can be regarded as equivalent to GC, highlighting the analytic nature of this formulation. Green coordinates produce angle-preserving deformation in 2D and quasi-conformal deformation in 3D. However, Green coordinates cannot simultaneously satisfy both given Dirichlet and Neumann boundary conditions. Moreover, angle preservation is a strict constraint; consequently, the resulting deformations often deviate significantly from the target cage. Furthermore, the deformation effects of GC are uniquely determined by the source and target cage positions, lacking additional parameter space to control shape preferences. Somigliana coordinates \cite{2023Somigliana}, derived from linear elasticity, feature normal control and incorporate additional parameters to control local shape area; however, the deformed shape may still deviate noticeably from the target cage. Another class of normal-controlled formulation, known as Biharmonic Coordinates (BiC)~\cite{2012Biharmonic, 2024Biharmonic3D, 2025PolyBiharmonic, 2025VariationalBiharmonic}, is established based on the biharmonic equation $\Delta^{2}u=0$ and utilizes normal derivatives as a new control dimension, providing a wide range of deformation control effects.
\section{Background}\label{sec:3}
As the mathematical background of our method, we first present basic concepts of Green’s function of a domain (connected open set) $\Theta \subseteq \mathbb{R}^2$, denoted as $G_{\Theta}(\xi, \eta)$. Detailed introductions to this mathematical background can be found in PDE textbooks such as~\citep{1995PDEbook}; here, however, we only introduce the parts relevant to our application. We first denote the well-known fundamental solution of the 2D Laplace equation as $\Phi(\xi, \eta)$, which has the following expression:
\begin{equation}
\label{eq:fundamental}
\Phi(\xi, \eta) = \frac{1}{2\pi}\log{\|\xi - \eta\|}.
\end{equation}
Then, the Green's function $G_{\Theta}(\xi, \eta)$ of a domain $\Theta \subseteq \mathbb{R}^2$ is the unique function satisfying the following properties:
\begin{enumerate}[label=(\roman*), leftmargin=*]
\item $G_{\Theta}(\cdot, \eta) - \Phi(\cdot, \eta)$ is harmonic on $\Theta$ and continuous on $\overline{\Theta}$. \label{prop:harmonic}
\item $G_{\Theta}(\xi, \eta)=0$ for $\eta \in \Theta$ and $\xi\in \partial \Theta$. \label{prop:zero}
\end{enumerate}
Here, $G_{\Theta}(\cdot, \eta)$ denotes that we regard $G_{\Theta}$ as a function of $\xi$ for each fixed $\eta$. The above definition indicates that $G_{\Theta}(\cdot, \eta) - \Phi(\cdot, \eta)$ is the solution of the Dirichlet problem
\begin{equation}
\Delta w(\xi)=0 \ \mathrm{on} \ \Theta,\ \ w(\xi)=-\Phi(\xi, \eta) \ \mathrm{on} \ \partial \Theta
\end{equation}
and therefore unique. Since $G_{\Theta}(\xi, \eta)-\Phi(\xi, \eta)$ is harmonic on $\Theta$ with respect to $\xi$, it follows that $G_{\Theta}(\xi, \eta)$ satisfies $\Delta_{\xi} G_{\Theta}(\xi, \eta)=\delta(\xi-\eta)$ for $\xi, \eta \in \Theta$. However, $G_{\Theta}(\xi,\eta)$ is defined on $\overline{\Theta} \times \Theta$, and $G_{\Theta}(\xi,\eta)$ vanishes when $\xi \in \partial\Theta$, while $\Phi(\xi, \eta)$ is defined on the whole $\mathbb{R}^2 \times \mathbb{R}^2$.

Not all regions $\Theta$ admit an analytic solution for their Green’s function. However, if an analytic conformal map exists that maps $\Theta$ onto the unit disk or the upper half-plane, then the Green’s function for that region possesses an analytic solution. A commonly used case is when $\Theta$ is a disk of radius $R$, where the Green’s function admits an analytic solution. Furthermore, when $\Theta$ is a two-dimensional rectangular region, the Green’s function does not admit an analytic solution, but it admits a semi-analytic expression and can be approximated by a series summation~\cite{2006GreenRectangle}.

In the following, we provide the analytic solution for the case where $\Theta$ is a disk, as we will subsequently use it to achieve various deformation effects in our method. We also provide its derivation based on the method of images in Appendix C of the supplementary material. Furthermore, we directly present the semi-analytical expression for the Green’s function of a rectangular domain; for detailed computational methods, please refer to the mathematics literature~\cite{2006GreenRectangle}. For source cages with a significant discrepancy between their width and height, modeling $\Theta$ with a rectangular domain can better approximate the source cage, and thus enable more diverse deformation effects.

Let $B_R(\mathbf{0})\subset\mathbb{R}^2$ be the open disk centered at the origin with radius $R$. The analytic solution of the Green's function of $B_R(\mathbf{0})$ can be expressed as:
\begin{equation}
\label{eq:G_R}
\begin{aligned}
G_{R}(\xi, \eta) &= \frac{1}{2\pi}\log{\|\xi - \eta\|} - \frac{1}{4\pi}\log{\bigl(\|\xi\|^2\|\eta\|^2 - 2R^{2} (\xi \cdot \eta) + R^{4}\bigr)} \\
&+ \frac{1}{2\pi} \log {R}.
\end{aligned}
\end{equation}
We can derive the derivative of $G_{R}(\xi, \eta)$ directly as follows:
\begin{equation}
\label{eq:grad_G_R}
\nabla_\xi G_R(\xi,\eta) = \frac{1}{2\pi} \left( \frac{\xi - \eta}{\|\xi - \eta\|^2} - \frac{\|\eta\|^2 \xi - R^2 \eta}{\|\xi\|^2 \|\eta\|^2 - 2R^2(\xi\cdot\eta) + R^4} \right),
\end{equation}
which will be used later when establishing the domain-generalized Green’s third identity and domain-varying Green coordinates.

In addition to the disk domain, we will also use the semi-analytical expression of the Green’s function for the rectangular domain. Consider a rectangle $\Theta$ with $-w/2 \le x \le w/2, -h/2 \le y \le h/2$. The Green's function $G_{\Theta}(\xi, \eta)$ can be expressed as a double Fourier sine series. We denote $\xi=(\xi_1, \xi_2)^{\top}, \eta=(\eta_1,\eta_2)^{\top}$, and $\overline{\xi}_1=\xi_1+w/2,\overline{\eta}_1=\eta_1+w/2,\overline{\xi}_2=\xi_2+h/2, \overline{\eta}_2=\eta_2+h/2$. Then,
\begin{align}
\label{eq:rec_G}
G_{\Theta}(\xi, \eta) = \frac{4}{wh}\sum_{m=1}^{\infty}\sum_{n=1}^{\infty}
\frac{\sin\frac{m\pi\overline{\xi}_1}{w}\,\sin\frac{m\pi\overline{\eta}_1}{w}\;
      \sin\frac{n\pi\overline{\xi}_2}{h}\sin\frac{n\pi\overline{\eta}_2}{h}}
     {\left(\frac{m\pi}{w}\right)^2+\left(\frac{n\pi}{h}\right)^2}.
\end{align}
Additionally, the gradient of $G_{\Theta}(\xi_1,\xi_2;\eta_1,\eta_2)$ with respect to $\xi=(\xi_1, \xi_2)^{\top}$ can be expressed as 
$\nabla_{\xi}G_{\Theta}(\xi, \eta)=(\frac{\partial G_{\Theta}}{\partial \xi_1}, \frac{\partial G_{\Theta}}{\partial \xi_2})^{\top}$,
where
\begin{align}
\frac{\partial G_{\Theta}}{\partial \xi_1}= \frac{4\pi}{w^2 h}\sum_{m=1}^{\infty}\sum_{n=1}^{\infty}
\frac{m\cos\frac{m\pi \overline{\xi}_1}{w}\sin\frac{m\pi\overline{\eta}_1}{w}
      \sin\frac{n\pi\overline{\xi}_2}{h}\sin\frac{n\pi\overline{\eta}_2}{h}}
     {\left(\frac{m\pi}{w}\right)^2+\left(\frac{n\pi}{h}\right)^2},
\end{align}
\begin{align}
\frac{\partial G_{\Theta}}{\partial \xi_2}= \frac{4\pi}{w h^2}\sum_{m=1}^{\infty}\sum_{n=1}^{\infty}
\frac{n\sin\frac{m\pi \overline{\xi}_1}{w}\sin\frac{m\pi\overline{\eta}_1}{w}
      \cos\frac{n\pi\overline{\xi}_2}{h}\sin\frac{n\pi\overline{\eta}_2}{h}}
     {\left(\frac{m\pi}{w}\right)^2+\left(\frac{n\pi}{h}\right)^2}.
\end{align}
\section{Method}
Based on the background introduced previously, this section establishes the deformation formulation using 2D Green’s functions of diverse domains. We denote the domain of the Green’s function as $\Theta$, while the region enclosed by the cage is $\Omega$. Note that these are two independent sets satisfying $\Omega \subseteq \Theta \subseteq \mathbb{R}^2$. Fig.~\ref{fig:teaser} clearly illustrates these two regions as well as the relationship between the object, the domain and the cage.  

\subsection{Domain-Varying Green Coordinates (DVGC)}
Green’s identity is a well-known result in mathematics; furthermore, Green’s third identity can be employed to establish traditional Green coordinates~\cite{2008Green, 2009DriveGreen}. Our main observation is that the derivation of the classical Green’s third identity is primarily based on the property that $\Delta_{\xi}\Phi(\xi, \eta)=\delta(\xi-\eta)$, where $\Phi(\xi, \eta)$ denotes the fundamental solution of the Laplace equation. Consequently, when the domain $\Theta$ varies, since $G_{\Theta}(\xi, \eta) - \Phi(\xi, \eta)$ is harmonic with respect to $\xi$ on $\Theta$, it follows that $G_{\Theta}(\xi, \eta)$ also satisfies this condition, i.e., $\Delta_{\xi} G_{\Theta}(\xi, \eta)=\delta(\xi-\eta)$. Consequently, we can derive a corresponding boundary integral expression based on $G_{\Theta}(\xi, \eta)$ analogous to the classical Green’s third identity. Using this formula, we can define deformation coordinates similar to the classical Green coordinates. Furthermore, different domains $\Theta$ yield different deformation effects. We therefore refer to our method as Domain-Varying Green Coordinates (DVGC). We begin with the following domain-generalized Green’s identity as our theoretical foundation.
\begin{theorem}[Domain-generalized Green's Third Identity]
Let $\Theta \subseteq \mathbb{R}^{2}$ be a domain, and denote by $G_{\Theta}(\xi, \eta)$ the Green's function of $\Theta$. Moreover, let $\Omega \subseteq \Theta$ be a bounded open set contained in $\Theta$ with $\mathcal{C}^{2}$ boundaries. $u(\xi) \in \mathcal{C}^{2}(\overline{\Omega})$ is a harmonic function, i.e., $\Delta u = 0$ in $\Omega$. Then, for $\eta \in \Omega$, we have:
\begin{equation}
\label{eq:third_identity}
u(\eta)=\int_{\partial \Omega}{u(\xi) \frac{\partial G_{\Theta}}{\partial \mathbf{n}}(\xi, \eta) \ \mathrm{d} \sigma_{\xi}} -\int_{\partial \Omega}{G_{\Theta}(\xi,\eta)\frac{\partial u}{\partial \mathbf{n}}(\xi) \ \mathrm{d} \sigma_{\xi}}.
\end{equation}
The directional derivatives are taken along the outward normal. 
\end{theorem}
The detailed proof is provided in Appendix A of the supplementary material. Although Walk on Stars~\citep{2023WalkStars} also employs a BEM formulation built on Green’s functions across different domains to establish the Monte Carlo PDE solver, our method strictly proves the generalized Green’s identity and specifies its application in cage-based deformation. Using Eq.~\eqref{eq:third_identity}, we derive the deformation formula on oriented simplicial cages $P=(\mathbb{V}, \mathbb{T})$ by substituting $u(\eta)=\eta$ into Eq.~\eqref{eq:third_identity}, and discretizing the above formulation based on cage vertices $\{\mathbf{v}_i|i \in I_{\mathbb{V}}\}$ and normals $\{\mathbf{n}_j|j\in I_{\mathbb{T}}\}$, yielding natural linear reproduction properties as follows:
\begin{equation}
\label{eq:isotropic_GC_exp}
\eta = F(\eta, P)=\sum_{i \in I_{\mathbb{V}}}{\phi_{i}(\eta) \mathbf{v}_i} + \sum_{j \in I_{\mathbb{T}}}{\psi_{j}(\eta) \mathbf{n}_j}, \ \ \eta \in \Omega,
\end{equation}
where the coordinate functions $\phi_{i}(\eta)$ and $\psi_{j}(\eta)$ are expressed as:
\begin{align}
\phi_{i}(\eta) &= \int_{\xi \in \mathcal{N}\{\mathbf{v}_i\}}{\Gamma_{\mathbf{v}_i}(\xi) \frac{\partial G_{\Theta}}{\partial \mathbf{n}}(\xi, \eta) \ \mathrm{d} \sigma_{\xi} , i \in I_{\mathbb{V}}}, \label{eq:isotropic_phi}\\
\psi_{j}(\eta) &= -\int_{\xi \in f_j}{G_{\Theta}(\xi, \eta) \ \mathrm{d}\sigma_{\xi}, j\in I_{\mathbb{T}}}.\label{eq:isotropic_psi}
\end{align}
Here, $I_\mathbb{V}$ and $I_\mathbb{T}$ denote the index sets for the vertices and faces of the cage, respectively, and $\mathcal{N}\{\mathbf{v}_i\}$ refers to faces incident to the vertex $\mathbf{v}_i$. The term $\Gamma_{\mathbf{v}_i}(\xi)$ represents a piecewise linear hat function defined on $\mathcal{N}\{\mathbf{v}_i\}$, which takes the value of one at $\mathbf{v}_i$, zero at any other vertex, and varies linearly across each face associated with $\mathbf{v}_i$ (specifically, $f_{j-1}=\overrightarrow{\mathbf{v}_{j-1}\mathbf{v}_{j}}$ and $f_j=\overrightarrow{\mathbf{v}_j\mathbf{v}_{j+1}}$ in 2D). When the source cage $P$ is deformed into the target cage $\tilde{P}$, the updated position $\tilde{\eta}$ of $\eta$ in the deformed object is computed as:
\begin{equation}
\label{eq:deformed_formulation}
\tilde{\eta} = F(\eta, \tilde{P})=\sum_{i \in I_{\mathbb{V}}}{\phi_{i}(\eta) \tilde{\mathbf{v}}_i} + \sum_{j \in I_{\mathbb{T}}}{\psi_{j}(\eta) s_j\tilde{\mathbf{n}}_j},
\end{equation}
where $\tilde{\mathbf{v}}_i$ and $\tilde{\mathbf{n}}_j$ represent the vertices and normals of the deformed cage, respectively. Furthermore, $s_j=\|\tilde{f}_j\|/\|f_j\|$ denotes the ratio of the face length $f_j$ after deformation to that before deformation. This term is introduced to ensure the scale invariance of the deformation, as suggested by~\citet{2008Green}.
\subsection{Relations with Harmonic and Green Coordinates}\label{sec:4_2}
Harmonic Coordinates (HC) \cite{2006GeneralConstructBC, 2007Harmonic} and Green Coordinates (GC) \cite{2008Green, 2009DriveGreen} are well-established approaches in this field. Their deformation behaviors usually differ significantly because HC are formulated based solely on cage vertices, whereas GC incorporate cage normals. Moreover, HC typically lack a closed-form solution and requires solving the corresponding PDE using finite element or boundary element methods, whereas GC feature closed-form expressions. However, in this section, we demonstrate that they both belong to the DVGC we proposed, representing two special cases of DVGC. Specifically, when the Green's function domain $\Theta$ aligns with the cage domain $\Omega$, the deformation effects are mathematically equivalent to HC. Meanwhile, when $\Theta=\mathbb{R}^2$, the deformation effects are mathematically equivalent to GC. Therefore, as $\Theta$ gradually expands from $\Omega$ to $\mathbb{R}^2$, we may expect the deformation effects to transform from HC to GC, thereby achieving diverse deformation effects and introducing an additional parameter space to control the deformation from being ‘more consistent with the cage’ to ‘more shape-preserving’. We will propose the following two theorems to state this fact and provide rigorous proofs for them. We also provide further insights for the $\Theta=\mathbb{R}^2$ case in Appendix B of the supplementary material.

\begin{theorem}\label{thm:HC}
When $\Theta=\Omega$, DVGC coincide with HC. Specifically, the Neumann coordinates vanish $(\psi_{j}(\eta)\equiv 0)$, and the Dirichlet coordinates $\phi_i(\eta)=h_i({\eta})$, where $h_i(\eta)$ are harmonic coordinates satisfying
\begin{equation}
\Delta h_i=0 \ \text{in} \ \Omega,\qquad h_i|_{\partial\Omega}=\Gamma_{\mathbf{v}_i},
\end{equation}
where $\Gamma_{\mathbf{v}_i}$ is the piecewise linear hat function on the cage boundary.
\end{theorem}
\begin{proof}
According to Property~\ref{prop:zero} of the Green's function introduced in Section~\ref{sec:3}, we have $G_\Omega(\xi,\eta)=0$ for $\xi\in\partial\Omega$ when $\
\Theta=\Omega$. Therefore,
\begin{equation}
\psi_j(\eta)=-\int_{\xi \in f_j} G_\Omega(\xi,\eta) \ \mathrm d\sigma_\xi=0.
\end{equation}
For the Dirichlet term $\phi_{i}(\eta)$, applying Eq.~\eqref{eq:third_identity} to the harmonic function $h_i$ and notice that $G_\Omega(\xi,\eta)=0$ for $\xi \in \partial\Omega$, we obtain
\begin{equation}
h_i(\eta)=\int_{\partial\Omega} h_i(\xi)\frac{\partial G_\Omega}{\partial\mathbf n}(\xi,\eta) \ \mathrm d\sigma_\xi
=\int_{\partial\Omega}\Gamma_{\mathbf{v}_i}(\xi)\frac{\partial G_\Omega}{\partial\mathbf n}(\xi,\eta) \ \mathrm d\sigma_\xi
=\phi_i(\eta).
\end{equation}
The second equality follows from the boundary condition $h_i|_{\partial\Omega}=\Gamma_{\mathbf{v}_i}$~\citep{2007Harmonic}. Therefore, the deformation formula of DVGC reduces to
\begin{equation}
\eta=\sum_i h_i(\eta)\mathbf{v}_i.
\end{equation}
Based on the uniqueness of the Dirichlet problem of the Laplace equation, we proved that DVGC are equivalent to HC at this case. 
\end{proof}

\begin{theorem}\label{thm:GC}
Let $\Theta=B_R(\mathbf{0})$ be a disk centered at the origin and let $\Omega \subset B_R(\mathbf{0})$. As $R \to +\infty$, the deformed shape obtained by DVGC converges pointwise to that obtained by GC. Consequently, GC is the limiting case $\Theta=\mathbb{R}^2$ of DVGC.
\end{theorem}
\begin{proof}
According to Eqs.~\eqref{eq:G_R} and~\eqref{eq:grad_G_R}, when $R \to +\infty$,
\begin{equation}
\begin{aligned}
G_R(\xi,\eta)&=\Phi(\xi,\eta)-\frac{1}{4\pi}\log{(R^4(1+O(R^{-2})))} + \frac{1}{2\pi}\log{R} \\
&=\Phi(\xi,\eta)-\frac{1}{2\pi}\log{R}+O(R^{-2}),
\label{eq:convergence_G_R}
\end{aligned}
\end{equation}
and
\begin{equation}
\frac{\partial G_R}{\partial\mathbf n}(\xi, \eta)=\frac{\partial \Phi}{\partial\mathbf n}(\xi, \eta)+O(R^{-2}).
\label{eq:convergence_pG_R_pn}
\end{equation}
The Dirichlet term $\phi_i^{(R)}(\eta)$ and Neumann term $\psi_j^{(R)}(\eta)$ of DVGC with $\Theta=B_R(\mathbf{0})$ can be expressed as
\begin{align}
\phi_i^{(R)}(\eta) &= \int_{\xi \in \mathcal{N}\{\mathbf{v}_i\}}\Gamma_{\mathbf{v}_i}(\xi)\frac{\partial G_R}{\partial\mathbf n}(\xi,\eta) \ \mathrm{d}\sigma_{\xi},\\
\psi_j^{(R)}(\eta) &= -\int_{\xi \in f_j}G_R(\xi,\eta) \ \mathrm{d}\sigma_{\xi}.
\end{align}
The corresponding coordinates of the original GC are denoted as $\phi_i^{(GC)}(\eta)$ and $\psi_j^{(GC)}(\eta)$, respectively. According to Eq.~\eqref{eq:convergence_pG_R_pn}, we obtain
\begin{equation}\label{eq:phi_conv}
\lim_{R \to +\infty}\phi_i^{(R)}(\eta)=\phi_i^{(\mathrm{GC})}(\eta) \ \text{for} \ \eta \in \Omega.
\end{equation}
For the Neumann term, integrating Eq.~\eqref{eq:convergence_G_R} over $f_j$ yields
\begin{equation}
\label{eq:psi_asym}
\psi_j^{(R)}(\eta)=\psi_j^{(\mathrm{GC})}(\eta)+\frac{\|f_j\|}{2\pi}\log R+O(R^{-2}).
\end{equation}
The deformed position of DVGC is therefore given by
\begin{equation}
\tilde{\eta}^{(R)}=\sum_{i \in I_{\mathbb{V}}} \phi_i^{(R)}(\eta)\tilde{\mathbf{v}}_i+\sum_{j \in I_{\mathbb{T}}} \psi_j^{(R)}(\eta) s_j \tilde{\mathbf{n}}_j.
\end{equation}
Substituting~\eqref{eq:phi_conv} and~\eqref{eq:psi_asym} into this equation, we obtain
\small{
\begin{equation}
\tilde{\eta}^{(R)}
=\sum_{i \in I_{\mathbb{V}}} \phi_i^{(\mathrm{GC})}(\eta)\tilde{\mathbf v}_i
+\sum_{j \in I_{\mathbb{T}}} \psi_j^{(\mathrm{GC})}(\eta)s_j\tilde{\mathbf n}_j
+\frac{\log R}{2\pi}\sum_{j \in I_{\mathbb{T}}} \|f_j\|s_j\tilde{\mathbf n}_j
+O(R^{-2}).
\end{equation}}
Below, we demonstrate that the term $\sum_{j \in I_{\mathbb{T}}} \|f_j\|s_j\tilde{\mathbf{n}}_j
$ vanishes. We first observe that $\sum_{j \in I_{\mathbb{T}}} \|f_j\|s_j\tilde{\mathbf{n}}_j
=\sum_{j \in I_{\mathbb{T}}} \|\tilde{f}_j\|\tilde{\mathbf n}_j$ according to the definition $s_j=\|\tilde{f}_j\|/\|f_j\|$,
where $\|\tilde{f}_j\|$ is edge length of the deformed cage and $\tilde{\mathbf n}_j$ is the corresponding outward normal vector. Let $\tilde{\mathbf{t}}_j$ denote the unit vector of $\tilde{f}_j$ in the counterclockwise direction. We easily obtain $\sum_{j \in I_{\mathbb{T}}} \|\tilde{f}_j\|\tilde{\mathbf t}_j = \mathbf{0}$, as the target cage is closed. Since $\tilde{\mathbf n}_j=\tilde{\mathbf t}_j^{\perp}$, where $\perp$ represents a vector with 90-degree clockwise rotation, we have $\sum_{j \in I_{\mathbb{T}}} \|f_j\|s_j\tilde{\mathbf{n}}_j
=\sum_{j \in I_{\mathbb{T}}} \|\tilde{f}_j\|\tilde{\mathbf n}_j=\mathbf{0}$, and therefore
\begin{equation}
\lim_{R \to +\infty}\tilde{\eta}^{(R)}
=\sum_{i \in I_{\mathbb{V}}} \phi_i^{(\mathrm{GC})}(\eta)\tilde{\mathbf v}_i
+\sum_{j \in I_{\mathbb{T}}} \psi_j^{(\mathrm{GC})}(\eta)s_j\tilde{\mathbf n}_j,
\end{equation}
which is precisely the deformation formula of GC.
\end{proof}

Based on the analysis above, we can expect that as $\Theta$ transitions from $\Omega$ to $\mathbb{R}^2$, our method can generate a family of deformation effects distinct from both HC and GC. However, not all $\Theta$ possess an analytic Green’s function. Therefore, we employ the analytical and semi-analytical expressions introduced in Section~\ref{sec:3} for $\Theta$ being a disk and $\Theta$ being a rectangle, respectively. By substituting these into Eqs.~\eqref{eq:isotropic_phi} and~\eqref{eq:isotropic_psi}, we establish the deformation coordinates, and further can use Eq.~\eqref{eq:deformed_formulation} to compute the deformed shape. We will demonstrate in the next section that when $\Theta$ is a disk, our method possesses a closed-form solution and can therefore be calculated rapidly without numerical integration.

\subsection{Closed-form DVGC expression for disk Green’s functions}
This section provides a detailed derivation of the closed-form expression of DVGC for disk Green’s functions, i.e., $\Theta=B_R(\mathbf{0})$. In this case, the Green’s function $G_{R}(\xi, \eta)$ admits an analytic formulation, as illustrated in Eq.~\eqref{eq:G_R}. We further demonstrate that the deformation coordinates $\phi_{i}(\eta)$ and $\psi_{j}(\eta)$ also admit a closed-form solution. Specifically, the integrals in Eqs.~\eqref{eq:isotropic_phi} and \eqref{eq:isotropic_psi} possess explicit expressions and do not require numerical integration.

The 2D cage $\mathbb{P}=(\mathbb{V}, \mathbb{T})$ is defined as a discrete, closed, and oriented one-manifold composed of polygonal chains. The region bounded by $\mathbb{P}$ is denoted by $\Omega$, where $\Omega \subseteq B_R(\mathbf{0}) \subseteq \mathbb{R}^2$. We consider the evaluation of the integrals $\phi_{i}(\eta)$ and $\psi_{j}(\eta)$ over a single line segment $f_j=\overrightarrow{\mathbf{v}_j\mathbf{v}_{j+1}}$ of the cage, parametrized by
$\xi(t)=\mathbf{a}+t\mathbf{b}$ for $t\in[0,1]$, where $\mathbf{a}=\mathbf{v}_j,\mathbf{b}=\mathbf{v}_{j+1}-\mathbf{v}_j$. Consequently, $\mathrm{d}\sigma_{\xi} = \|\mathbf{b}\| \ \mathrm{d}t$. Fig.~\ref{fig:Figure2} presents an example of a 2D cage and its mathematical notations.
\begin{figure}[htb]
  \centering
  \includegraphics[width=0.6\linewidth]{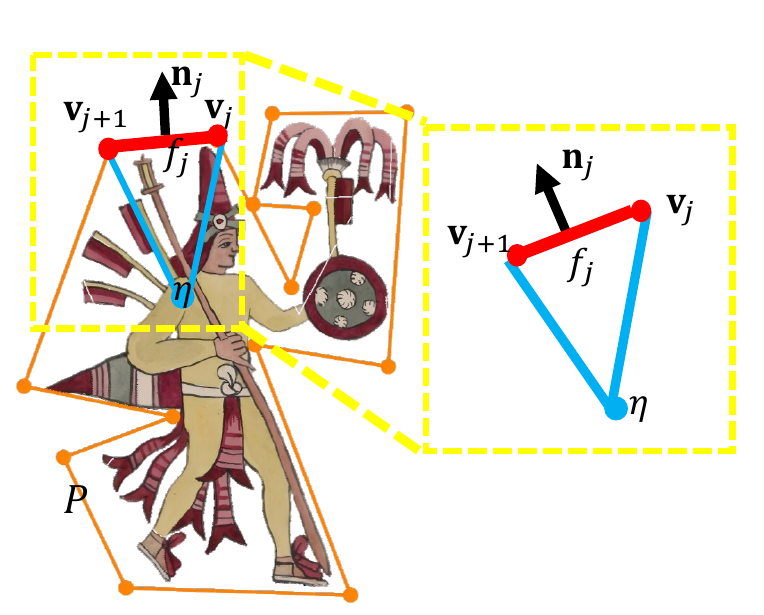}
  \caption{\label{fig:Figure2}
          The illustration of a 2D cage and its mathematical notations. }
\end{figure}

In the following, we will frequently use the following two anti-derivative expressions. Let $Q(t)$ be a quadratic polynomial defined as
\begin{equation}
Q(t)=\alpha t^2 + 2\beta t + \gamma, \qquad \alpha, \beta, \gamma \in \mathbb{R}.
\end{equation}
Furthermore, denote
\begin{equation}
\qquad s=\sqrt{\alpha \gamma - \beta^2}.
\end{equation}
Then,
\begin{align}
\int \log Q(t) \ \mathrm{d}t
&= \Big(t+\frac{\beta}{\alpha}\Big)\log Q(t)
-2\Big(t+\frac{\beta}{\alpha}\Big) \notag  \\
&+2\frac{s}{\alpha}\arctan\Big(\frac{\alpha t+\beta}{s}\Big) + C, \label{eq:anti_derivative1} \\
\int\frac{p t+q}{Q(t)}\ \mathrm{d}t
&= \frac{p}{2\alpha}\log Q(t)
+ \frac{q-\frac{p\beta}{\alpha}}{s}\arctan\Big(\frac{\alpha t+\beta}{s}\Big) + C.
\label{eq:anti_derivative2}
\end{align}
We transform the integration calculation of deformation coordinates into the aforementioned formula, and subsequently derive the closed-form expression. We first focus on the Neumann term $\psi_j(\eta)$, where
\begin{equation}
\psi_{j}(\eta)=-\int_{\xi \in f_j} G_R(\xi,\eta) \ \mathrm{d}\sigma_\xi
= -\|\mathbf{b}\|\int_{0}^{1} G_R(\mathbf{a}+t\mathbf{b},\eta) \ \mathrm{d}t.
\end{equation}
Define
\begin{equation}
\alpha_q = \|\mathbf{b}\|^2,\qquad
\beta_q = (\mathbf{a}-\eta) \cdot \mathbf{b},\qquad
\gamma_q = \|\mathbf{a}-\eta\|^2.
\end{equation}
Then
\begin{equation}
Q(t)=\|\xi(t)-\eta\|^2 = \alpha_q t^2 + 2\beta_q t + \gamma_q.
\end{equation}
Furthermore, define
\begin{equation}
\begin{gathered}
\alpha_s = \|\eta\|^2\|\mathbf{b}\|^2,\quad
\beta_s = \|\eta\|^2 (\mathbf{a}\cdot \mathbf{b}) - R^2 (\mathbf{b}\cdot\eta),\quad \\
\gamma_s = \|\eta\|^2 \|\mathbf{a}\|^2 - 2R^2(\mathbf{a}\cdot\eta) + R^4.
\end{gathered}
\end{equation}
Then
\begin{equation}
S(t)=\|\xi(t)\|^2\|\eta\|^2 - 2R^{2} (\xi(t) \cdot \eta) + R^{4}=\alpha_s t^2 + 2\beta_s t + \gamma_s,
\end{equation}
and we obtain
\begin{equation}
\psi_{j}(\eta)= -\|\mathbf{b}\|\int_0^1 (\frac{1}{4\pi} \log{Q(t)}-\frac{
1}{4\pi} \log{S(t)}+\frac{1}{2\pi} \log R) \ \mathrm{d}t.
\end{equation}
Using the antiderivative formula (Eq.~\eqref{eq:anti_derivative1}), and define
\begin{equation}
s_q = \sqrt{\alpha_q\gamma_q-\beta_q^2}, \quad s_s = \sqrt{\alpha_s\gamma_s-\beta_s^2},
\end{equation}
we obtain the closed-form solution of $\psi_j(\eta)$ as
\begin{equation}
\psi_j(\eta)
= -\frac{\|\mathbf{b}\|}{4\pi}[F_Q(1)-F_Q(0)]
+\frac{\|\mathbf{b}\|}{4\pi}[F_S(1)-F_S(0)]-\frac{\|\mathbf{b}\|}{2\pi} \log{R},
\end{equation}
where
\begin{equation}
\begin{aligned}
F_Q(t) &= \Big(t+\frac{\beta_q}{\alpha_q}\Big)\log Q(t)
-2\Big(t+\frac{\beta_q}{\alpha_q}\Big)
+2\frac{s_q}{\alpha_q}\arctan\Big(\frac{\alpha_q t+\beta_q}{s_q}\Big),\\
F_S(t) &= \Big(t+\frac{\beta_s}{\alpha_s}\Big)\log S(t)
-2\Big(t+\frac{\beta_s}{\alpha_s}\Big)
+2\frac{s_s}{\alpha_s}\arctan\Big(\frac{\alpha_s t+\beta_s}{s_s}\Big).
\end{aligned}
\end{equation}

Next, we focus on the calculation of the Dirichlet term $\phi_{\mathbf{v}_j}(\eta)$ for the vertex $\mathbf{v}_j$. This value is obtained by summing the contributions of its two associated edges. We concentrate on one of these edges and denote by $\phi_{\mathbf{v}_j, f_j}(\eta)$ the contribution of $f_j$ to $\phi_{j}(\eta)$. Then,
\begin{equation}
\begin{aligned}
\phi_{\mathbf{v}_j, f_j}(\eta)
=& \int_{\xi \in f_j} \Gamma_{\mathbf{v}_j, f_j}(\xi)\frac{\partial G_R}{\partial \mathbf{n}}(\xi,\eta) \ \mathrm{d}\sigma_\xi \\
=& \|\mathbf{b}\|\int_0^1 (1-t) \nabla_\xi G_R(\mathbf{a}+t\mathbf{b},\eta) \cdot \mathbf{n}_{j}\ \mathrm{d}t,
\end{aligned}
\end{equation}
where $\mathbf{n}_{j}$ is the outward normal of $f_j$. Given that $\mathbf{b} = \mathbf{v}_{j+1} - \mathbf{v}_{j}$, we have $\mathbf{b}\cdot \mathbf{n}_j=0$. Denoting
\begin{equation}
u_q=(\mathbf{a}-\eta) \cdot \mathbf{n}_j, \quad u_s=\|\eta\|^2 (\mathbf{a} \cdot \mathbf{n}_j) - R^2 (\eta \cdot \mathbf{n}_j).
\end{equation}
Then,
\begin{equation}
\phi_{\mathbf{v}_j, f_j}(\eta)=\frac{\|\mathbf{b}\|}{2\pi}(u_q\int_{0}^{1}{\frac{1-t}{Q(t)}} \ \mathrm{d}t-u_s\int_{0}^{1}{\frac{1-t}{S(t)}} \ \mathrm{d}t).
\end{equation}
Using the antiderivative formula (Eq.~\eqref{eq:anti_derivative2}) and denoting
\begin{equation}
\begin{aligned}
U_Q(t)
&= -\frac{u_q}{2\alpha_q}\log Q(t)
+ \frac{u_q\Big(1+\dfrac{\beta_q}{\alpha_q}\Big)}{s_q}\arctan\Big(\frac{\alpha_q t+\beta_q}{s_q}\Big),\\
U_S(t)
&= -\frac{u_s}{2\alpha_s}\log S(t)
+ \frac{u_s\Big(1+\dfrac{\beta_s}{\alpha_s}\Big)}{s_s}\arctan\Big(\frac{\alpha_s t+\beta_s}{s_s}\Big),
\end{aligned}
\end{equation}
we obtain
\begin{equation}
\phi_{\mathbf{v}_j, f_j}(\eta)
= \frac{\|\mathbf{b}\|}{2\pi}[U_Q(1)-U_Q(0)]-\frac{\|\mathbf{b}\|}{2\pi}[U_S(1)-U_S(0)].
\end{equation}
Then, we complete the derivation of the closed-form expression of DVGC for Green’s functions of disks.

\section{Experiments}
\begin{figure*}[htb]
  \centering
  \includegraphics[width=0.8\linewidth]{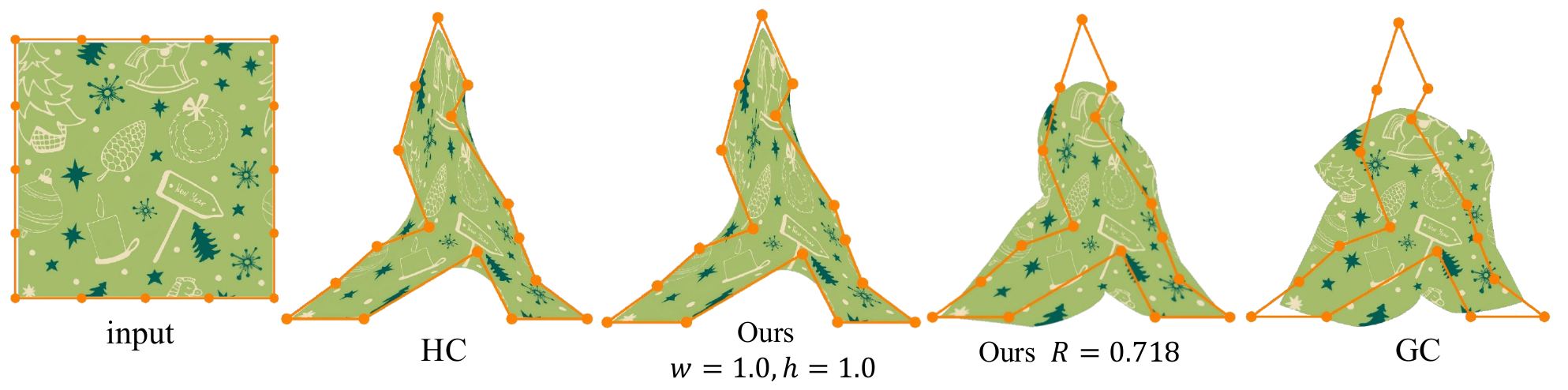}
  \caption{\label{fig:Figure3}
          Deformation results of HC and GC, as well as our method using minimal rectangular and disk Green’s function domains that enclose the cage.}
\end{figure*}
In this section, we present the experimental results of DVGC (our method) and compare with HC~\citep{2007Harmonic} and GC~\citep{2008Green}. All experiments are conducted on a standard laptop with an Intel CPU running at 2.40 GHz and 16GB of RAM. Section~\ref{sec:5_1} presents the results using Green’s functions of disks, while Section~\ref{sec:5_2} focuses on those using the Green’s function of rectangles. In Appendix D of the supplementary material, we provide further comparisons with Mean Value Coordinates (MVC)~\citep{2003MVCPoly} and Biharmonic Coordinates (BiC)~\citep{2012Biharmonic}.

\subsection{Deformation results using Green’s function of disks}\label{sec:5_1}

Fig.~\ref{fig:Figure4} shows the deformation results of HC and GC, as well as our method evaluated with disks of varying radii for the Green’s function domain, all using the same source and target cages. To ensure scale consistency, we first compute the axis-aligned bounding box of the cage, and normalize the cage and the input object such that the longer side of the width and height of the bounding box is scaled to $1.0$, with the origin placed at the center of the bounding box. We display the input object, source cage, target cage, and the results of HC, GC, and our proposed method for various inputs. For our method, we present results for three radii $R$ arranged in ascending order (left to right). The leftmost result of our method corresponds to the smallest radius enclosing the cage, while the other two illustrate the outcomes using progressively larger radii. It can be observed that the deformation results of HC align the target cage closely, while GC produce shape preserving effects and is strictly angle preserving in 2D. On the other hand, our method produces diverse deformation effects. As $R$ increases, the results gradually transition from HC to GC. Although the deformed shape of HC aligns closely with the target cage, while this seems desirable, the cage typically serves as a coarse simplicial representation of the input object. Consequently, the HC deformation often induces faceted creasing or distinct piecewise-linear effects. In contrast, GC can preserve some of the smooth structure of the original input; however, its deformation result may deviate noticeably from the target cage. Our method effectively balances these two characteristics. The cage of the moon-shaped example constitutes a non-convex region. When a convex Green’s function domain (disk) is employed, the areas near the convex boundaries approximate HC more closely, whereas areas near the concave boundaries remain close to GC. However, strict adherence to the cage is not always desirable. For instance, the HC results of this example suffer from significant shearing artifacts and alter the smiley face. In contrast, our method effectively preserves the shape and produces more natural deformation, while staying closer to the target cage than GC.

Table~\ref{table1} presents the running times (in seconds) of different approaches for the examples shown in Fig.~\ref{fig:Figure4}. Here, $|\mathcal{V}|$ denotes the number of vertices in the cage, and $|\mathcal{P}|$ represents the number of pixels requiring computation for the deformation of our method. HC are computed by discretizing the Laplace equation on the image grid using a five-point finite-difference scheme and solving the resulting sparse linear system, whereas GC and our method utilize a closed-form solution. The GC computation is a subset of our method, comprising only the first term of Eq.~\eqref{eq:G_R}. In contrast to HC, our method employs a boundary integral representation. 
Consequently, our method requires computation only at points of interest. For instance, in the sketch model of Exp.5 of the table, it is sufficient to compute solely the points on the sketch. In contrast, HC still necessitate finite element discretization of the region enclosed by the cage; thus, our method demonstrates a clear time advantage. For other examples, the time advantage of our method is not obvious because it requires computation for every pixel, whereas HC can directly implement finite element discretization using the image grid. Extending our method to 3D and deriving a closed-form solution in the future would confine calculations to the object vertices, which reside on a piecewise 2-manifold. In contrast, HC would still require finite element discretization of the entire 3D space. 

\begin{table}[t]
\centering
\caption{Running time (in seconds) of different approaches of the examples in Fig.~\ref{fig:Figure4}. Here, $|\mathcal{V}|$ denotes the number of vertices in the cage, and $|\mathcal{P}|$ represents the number of pixels requiring computation for the deformation.}
\label{table1}
\begin{tabular}{ccccccc}
\hline
Model & $|\mathcal{V}|$ & $|\mathcal{P}|$ & HC & GC & Ours \\
\hline

Exp. 1 & 13 & $1.14 \times 10^{6}$ & 5.40 & 2.99 & 7.97 \\
Exp. 2 & 26 & $4.14 \times 10^{5}$ & 2.40 & 1.98 &  5.81 \\
Exp. 3 & 24 & $5.54 \times 10^{5}$ & 3.27 &  2.42 & 7.16\\
Exp. 4 & 16 & $2.37 \times 10^{5}$ & 1.46 & 0.70 & 1.21\\
Exp. 5 & 16 & $7.22 \times 10^{4}$ & 1.32 & 0.14 & 0.34\\
Exp. 6 & 7 & $2.86 \times 10^{5}$ & 1.89 & 0.36 & 1.11 \\

\hline
\end{tabular}
\label{table_MAP}
\end{table}

\subsection{Deformation results using Green’s function of rectangles}\label{sec:5_2}
Disk Green’s functions are not always suitable for all source cages, especially when the width and height of the bounding box differ significantly. In such cases, adopting the smallest radius $R$ yields results that remain very similar to GC. Therefore, we employ the semi-analytical expression of the Green’s function for a rectangular region introduced in Eq.~\eqref{eq:rec_G} to achieve diverse deformation effects in this case. We take $m=n=500$ in Eq.~\eqref{eq:rec_G} and use the Gauss-Legendre quadrature with 96 integration points to compute $\phi_{i}(\eta)$ and $\psi_j(\eta)$ for DVGC in Eqs.~\eqref{eq:isotropic_phi} and~\eqref{eq:isotropic_psi}. We also add a small Gaussian kernel to the coordinates of neighboring pixels to avoid aliasing. Numerical integration completes calculations within an acceptable time (the longest example of Fig.~\ref{fig:Figure7} takes approximately 31 seconds). We normalize the cage and the input object such that the longer side of the cage bounding box is scaled to $1.0$. Fig.~\ref{fig:Figure3} shows the deformation results of HC, GC, as well as our method using minimal rectangular and disk Green’s function domains that enclose the cage. It can be observed that the results of HC and our method with the minimum rectangle yield almost identical effects, as this example satisfies $\Omega=\Theta$. This experimentally validates the theoretical correctness analyzed in Section~\ref{sec:4_2}.
In Fig.~\ref{fig:Figure7}, we mainly focus on the rectangle Green’s function and demonstrate two results of our method in each example: one uses the smallest rectangle enclosing the cage, and the other uses a slightly larger rectangle as the Green’s function domain. Our method yields deformation effects distinct from those of HC and GC. In the tower model, HC create a large sharp shape at the cage bends, while GC remain far from the target cage. Our method successfully avoids the weaknesses of both approaches. For the fish model, when $w=1.0, h=0.111$, it also satisfies $\Omega=\Theta$, further verifying the correctness of our theory.

\section{Conclusion}
In this work, we propose Domain-Varying Green Coordinates (DVGC) for 2D cage-based deformation and provide a unified framework that bridges Harmonic Coordinates (HC) and Green Coordinates (GC) using a domain-generalized Green’s third identity. When the domain $\Theta$ of the Green's function expands from the cage-enclosed region $\Omega$ to the full plane $\mathbb{R}^2$, the deformation effect continuously transitions from HC to GC. Consequently, we offer a new control space for the deformation effect, enabling it to be either better aligned with the target cage or more shape-preserving.

Nevertheless, our method has several limitations. The current implementation is restricted to 2D. While the underlying theory can be fully extended to 3D scenarios, deriving closed-form expressions for $\phi_{i}(\eta)$ and $\psi_{j}(\eta)$ is non-trivial for triangle mesh cages. Another limitation is the lack of closed-form expressions when the domain is not a disk. Although certain cases, such as rectangles, possess semi-analytical expressions, the calculation remains approximative and relies on numerical integration to compute the deformation coordinates, leading to trade-offs between time and accuracy. Nevertheless, we believe that we establish a theoretical framework for taking into account the domain of Green’s functions. Given that extensive explorations of Green's functions have been conducted in the mathematical field, we believe that our method paves the way for future research.

\bibliographystyle{ACM-Reference-Format}
\bibliography{sample-base}

\begin{figure*}[htb]
  \centering
  \includegraphics[width=1.0\linewidth]{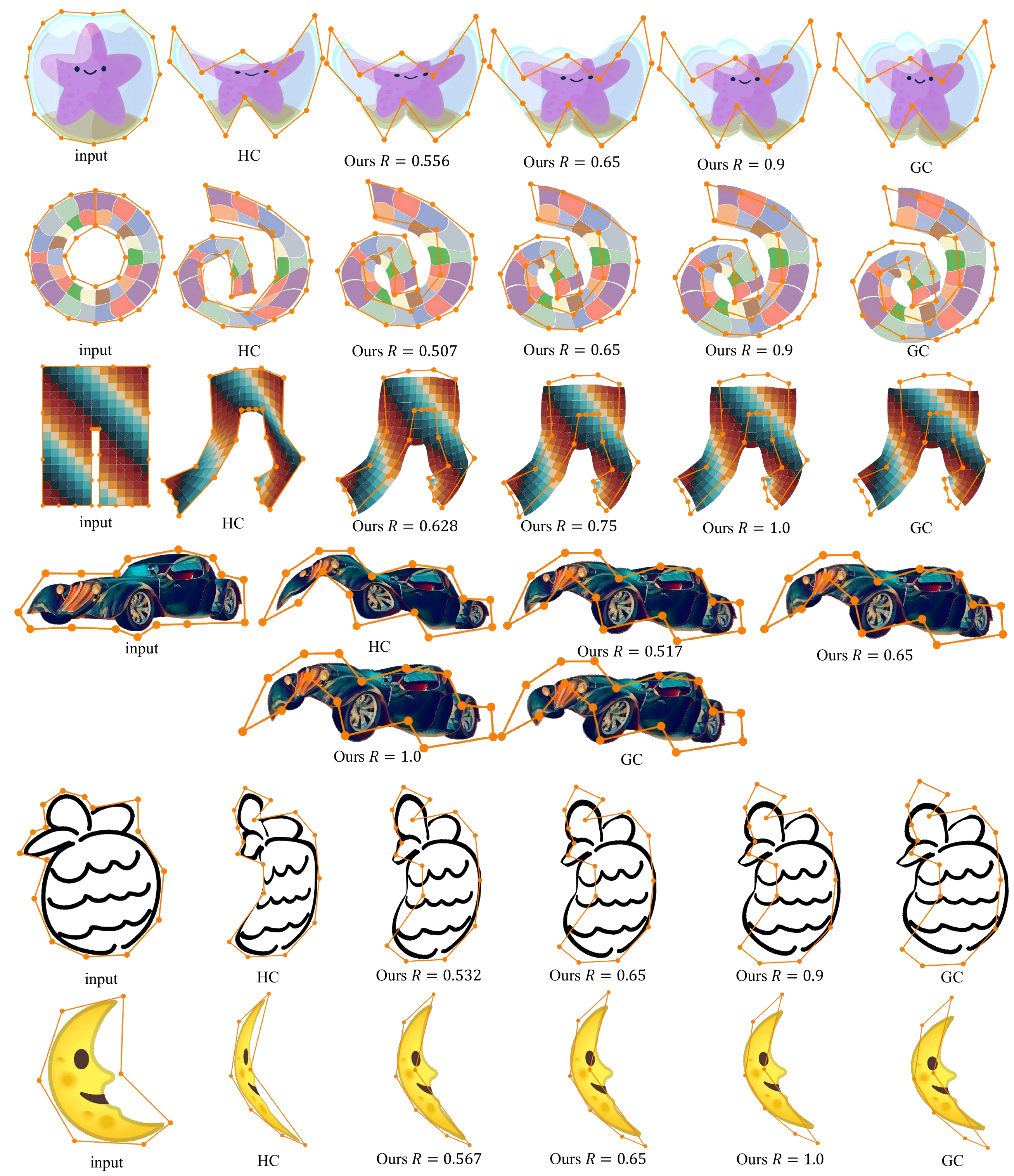}
  \caption{\label{fig:Figure4}
          Deformation results of HC and GC, as well as our method with different radii $R$ of disk Green's functions. The leftmost result of our method corresponds to the smallest radius enclosing the cage.}
\end{figure*}

\begin{figure*}[htb]
  \centering
  \includegraphics[width=1.0\linewidth]{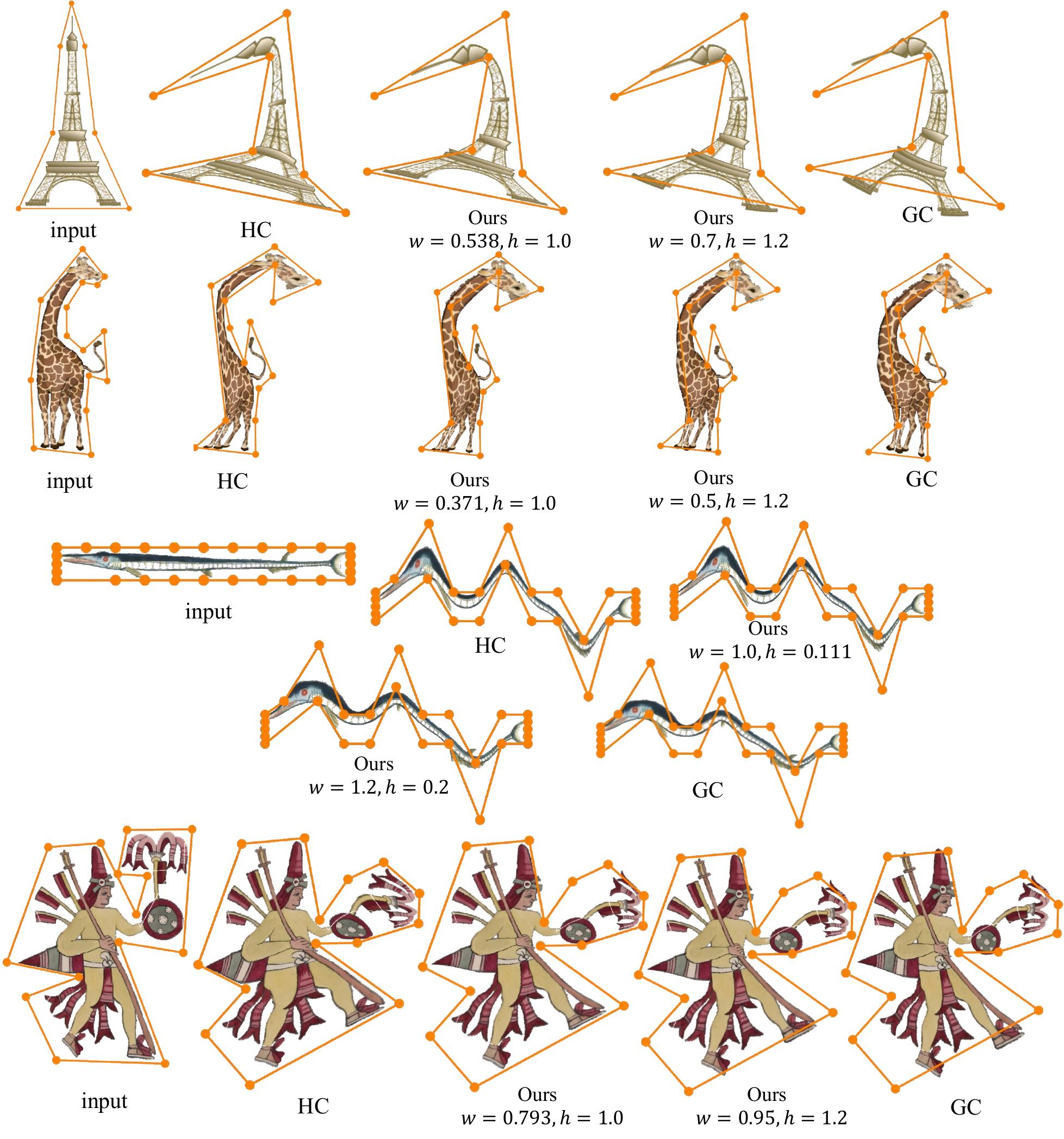}
  \caption{\label{fig:Figure7}
          Deformation results of HC and GC, as well as our method with different height $h$ and width $w$ of rectangle Green's functions. The leftmost result of our method corresponds to the smallest rectangle enclosing the cage.}
\end{figure*}


\clearpage


\section*{Supplementary material}

# Supplementary Material of *Domain-Varying 2D Green's Functions for Cage-based Deformation*

DONG XIAO, University of Science and Technology of China, China
RENJIE CHEN*, University of Science and Technology of China, China
BAILIN DENG, Cardiff University, UK

## A Domain-generalized Green's Third Identity

Theorem 1 (Domain-generalized Green's Third Identity). *Let $\Theta \subseteq \mathbb{R}^2$ be a domain, and denote by $G_\Theta(\xi, \eta)$ the Green's function of $\Theta$. Moreover, let $\Omega \subseteq \Theta$ be a bounded open set contained in $\Theta$ with $C^2$ boundaries. $u(\xi) \in C^2(\overline{\Omega})$ is a harmonic function, i.e., $\Delta u = 0$ in $\Omega$. Then, we have:*

$$u(\eta) = \int_{\partial\Omega} u(\xi) \frac{\partial G_\Theta}{\partial \mathbf{n}}(\xi, \eta) \, \mathrm{d}\sigma_\xi - \int_{\partial\Omega} G_\Theta(\xi, \eta) \frac{\partial u}{\partial \mathbf{n}}(\xi) \, \mathrm{d}\sigma_\xi. \tag{1}$$

*The directional derivatives are taken along the outward normal.*

Proof. We first apply the divergence theorem. Specifically, let $\mathbf{F} \in C^1(\overline{\Omega})$ be a vector-valued function. Then,

$$\int_\Omega \nabla \cdot \mathbf{F} \, \mathrm{d}V = \int_{\partial\Omega} \mathbf{F} \cdot \mathbf{n} \, \mathrm{d}\sigma, \tag{2}$$

where $\mathbf{n}$ is the outward normal of $\partial\Omega$. Let $\mathbf{F} = u\nabla v$, where $u, v \in C^2(\overline{\Omega})$ are two arbitrary scalar functions. By applying the rule for computing the divergence of function products, we obtain the following equation:

$$\nabla \cdot \mathbf{F} = \nabla u \cdot \nabla v + u\Delta v. \tag{3}$$

Combining Eqs. (2) and (3), we get

$$\int_\Omega (\nabla u \cdot \nabla v + u\Delta v) \, \mathrm{d}V = \int_{\partial\Omega} u \frac{\partial v}{\partial \mathbf{n}} \, \mathrm{d}\sigma. \tag{4}$$

Exchanging $u$ and $v$ in the above equation and obtain

$$\int_\Omega (\nabla v \cdot \nabla u + v\Delta u) \, \mathrm{d}V = \int_{\partial\Omega} v \frac{\partial u}{\partial \mathbf{n}} \, \mathrm{d}\sigma. \tag{5}$$

Then, we can subtract Eq. (5) from Eq. (4) and obtain

$$\int_\Omega (u\Delta v - v\Delta u) \, \mathrm{d}V = \int_{\partial\Omega} (u\frac{\partial v}{\partial \mathbf{n}} - v\frac{\partial u}{\partial \mathbf{n}}) \, \mathrm{d}\sigma. \tag{6}$$

Equation (6) can also refer to the Green's second identity. Based on this insight, we prove our conclusion in two approaches: one relies on distribution theory [Schwartz 1966], and the other adopts the $\varepsilon - \delta$ language and limit theory.

Let $v(\xi) = G_\Theta(\xi, \eta)$, where $\eta \in \Omega$ is fixed and the integral is taken with respect to $\xi$. Since $G_\Theta(\cdot, \eta) - \Phi(\cdot, \eta)$ is harmonic on $\Theta$ and continuous on $\overline{\Theta}$, and $\Omega \subseteq \Theta$, we have $\Delta v(\xi) = \delta(\xi - \eta)$. Based on the properties of the delta

*Corresponding author

Authors' Contact Information: Dong Xiao, xiaodong@ustc.edu.cn, School of Mathematical Sciences, University of Science and Technology of China, China; Renjie Chen, renjiec@ustc.edu.cn, School of Mathematical Sciences, University of Science and Technology of China, China; Bailin Deng, dengb3@cardiff.ac.uk, School of Computational and Mathematical Sciences, Cardiff University, UK.

distribution, we obtain

$$\int_{\Omega} u(\xi)\Delta G_{\Theta}(\xi,\eta)\ \mathrm{d}V_{\xi} = \int_{\Omega} u(\xi)\delta(\xi-\eta)\ \mathrm{d}V_{\xi} = u(\eta). \tag{7}$$

This completes the proof when $u$ satisfies $\Delta u = 0$ on $\Omega$.

We can also prove this using the $\varepsilon - \delta$ language, which avoids the singularity in the volume integral and therefore does not require the distribution theory, but it will be slightly more complicated. Let $\eta \in \Omega$ be fixed. Since $\eta$ is an interior point of $\Omega$, there exists $\varepsilon > 0$ such that $\overline{B_{\varepsilon}(\eta)} \subsetneq \Omega$, where $B_{\varepsilon}(\eta)$ is a disk centered at $\eta$ with radius $\varepsilon$, and $\overline{B_{\varepsilon}(\eta)}$ its closure. Let $\Omega_{\varepsilon} = \Omega \setminus \overline{B_{\varepsilon}(\eta)}$. On $\Omega_{\varepsilon}$, $v(\xi) = G_{\Theta}(\xi,\eta)$ is smooth and harmonic with respect to $\xi$. Since $u(\xi)$ is also

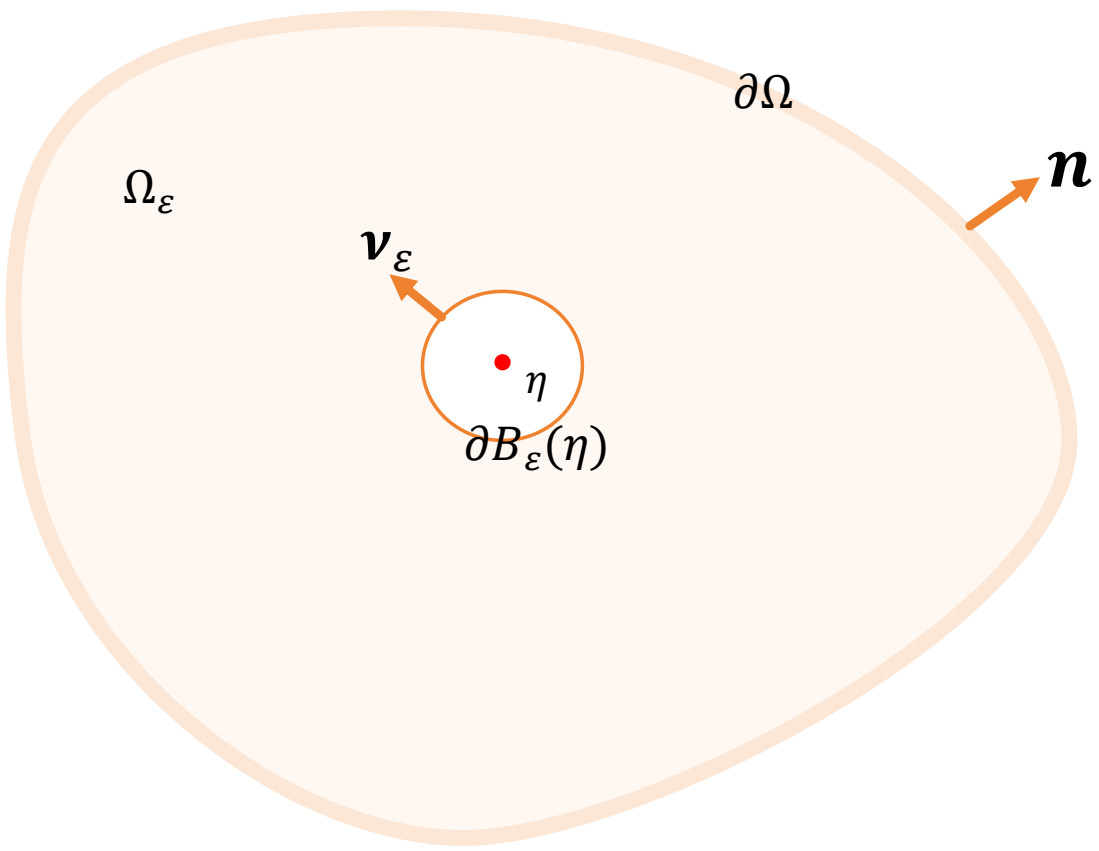


Fig. 1. Illustation of $\Omega_{\varepsilon}$, which is obtained by removing $B_{\varepsilon}(\eta)$ from $\Omega$.

harmonic on $\Omega$ based on the assumption, the left hand side of Eq. (6) is zero and we obtain

$$\int_{\partial\Omega_{\varepsilon}} \left( u(\xi)\frac{\partial v}{\partial \mathbf{n}_{\varepsilon}}(\xi) - v(\xi)\frac{\partial u}{\partial \mathbf{n}_{\varepsilon}}(\xi) \right)\ \mathrm{d}\sigma_{\xi} = 0, \tag{8}$$

where $\mathbf{n}_{\varepsilon}$ denotes the outward unit normal of $\Omega_{\varepsilon}$. On $\partial\Omega$ we have $\mathbf{n}_{\varepsilon} = \mathbf{n}$, while on $\partial B_{\varepsilon}(\eta)$ we have $\mathbf{n}_{\varepsilon} = -\boldsymbol{\nu}_{\varepsilon}$, where $\boldsymbol{\nu}_{\varepsilon}$ is the outward unit normal of $B_{\varepsilon}(\eta)$. See Fig. 1 for the illustration of $\Omega_{\varepsilon}$, $\mathbf{n}$ and $\boldsymbol{\nu}_{\varepsilon}$ .Hence

$$\int_{\partial\Omega} \left( u(\xi)\frac{\partial G_{\Theta}}{\partial \mathbf{n}}(\xi,\eta) - G_{\Theta}(\xi,\eta)\frac{\partial u}{\partial \mathbf{n}}(\xi) \right)\ \mathrm{d}\sigma_{\xi} = \int_{\partial B_{\varepsilon}(\eta)} u(\xi)\frac{\partial G_{\Theta}}{\partial \boldsymbol{\nu}_{\varepsilon}}(\xi,\eta)\ \mathrm{d}\sigma_{\xi} - \int_{\partial B_{\varepsilon}(\eta)} G_{\Theta}(\xi,\eta)\frac{\partial u}{\partial \boldsymbol{\nu}_{\varepsilon}}(\xi)\ \mathrm{d}\sigma_{\xi}. \tag{9}$$

We now evaluate the limit of the two integrals over $\partial B_{\varepsilon}(\eta)$ as $\varepsilon \longrightarrow 0$. The Green's function $G_{\Theta}(\xi,\eta)$ can be written as

$$G_{\Theta}(\xi,\eta) = \Phi(\xi,\eta) + H_{\Theta}(\xi,\eta), \tag{10}$$

where $\Phi(\xi,\eta) = \frac{1}{2\pi}\log\|\xi-\eta\|$ is the fundamental solution of Laplace equation, and $H_{\Theta}(\cdot,\eta)$ is harmonic in a neighborhood of $\eta$ and therefore is smooth there. On $\partial B_{\varepsilon}(\eta)$,

$$\frac{\partial\Phi}{\partial\boldsymbol{\nu}_{\varepsilon}}(\xi,\eta) = \nabla_{\xi}\Phi(\xi,\eta)\cdot\boldsymbol{\nu}_{\varepsilon} = \frac{(\xi-\eta)}{2\pi\|\xi-\eta\|^2}\cdot\frac{(\xi-\eta)}{\varepsilon} = \frac{1}{2\pi\varepsilon}. \tag{11}$$

Thus

$$\int_{\partial B_{\varepsilon}(\eta)} u(\xi)\frac{\partial\Phi}{\partial\boldsymbol{\nu}_{\varepsilon}}(\xi,\eta)\ \mathrm{d}\sigma_{\xi} = \frac{1}{2\pi\varepsilon}\int_{\partial B_{\varepsilon}(\eta)} u(\xi)\ \mathrm{d}\sigma_{\xi}. \tag{12}$$

We decompose the above integral as

$$\frac{1}{2\pi\varepsilon}\int_{\partial B_\varepsilon(\eta)} u(\xi)\,\mathrm{d}\sigma_\xi = u(\eta) + \frac{1}{2\pi\varepsilon}\int_{\partial B_\varepsilon(\eta)} [u(\xi) - u(\eta)]\,\mathrm{d}\sigma_\xi, \tag{13}$$

where the first term follows immediately from $|\partial B_\varepsilon(\eta)| = 2\pi\varepsilon$. For the second term, the continuity of $u$ at $\eta$ implies that for any $\delta > 0$, taking $\varepsilon$ sufficiently small gives $|u(\xi) - u(\eta)| < \delta$ on $\partial B_\varepsilon(\eta)$. Hence, its absolute value is bounded by

$$\frac{1}{2\pi\varepsilon}\cdot\delta\cdot|\partial B_\varepsilon(\eta)| = \delta. \tag{14}$$

By the arbitrariness of $\delta$, we obtain

$$\lim_{\varepsilon\to 0}\frac{1}{2\pi\varepsilon}\int_{\partial B_\varepsilon(\eta)} [u(\xi) - u(\eta)]\,\mathrm{d}\sigma_\xi = 0. \tag{15}$$

Therefore,

$$\lim_{\varepsilon\to 0}\int_{\partial B_\varepsilon(\eta)} u(\xi)\frac{\partial\Phi}{\partial \boldsymbol{v}_\varepsilon}(\xi,\eta)\,\mathrm{d}\sigma_\xi = u(\eta). \tag{16}$$

Moreover, since $H_\Theta(\cdot,\eta)$ is harmonic near $\eta$, $\partial H_\Theta/\partial \boldsymbol{v}_\varepsilon$ is bounded on $\partial B_\varepsilon(\eta)$ for sufficiently small $\varepsilon$. Hence

$$\left|\int_{\partial B_\varepsilon(\eta)} u(\xi)\frac{\partial H_\Theta}{\partial \boldsymbol{v}_\varepsilon}(\xi,\eta)\,\mathrm{d}\sigma_\xi\right| \le \int_{\partial B_\varepsilon(\eta)} |u(\xi)|\left|\frac{\partial H_\Theta}{\partial \boldsymbol{v}_\varepsilon}(\xi,\eta)\right|\,\mathrm{d}\sigma_\xi \le C_1C_2|\partial B_\varepsilon(\eta)| = O(\varepsilon) \to 0 \text{ as } \varepsilon\to 0. \tag{17}$$

Since $G_\Theta(\xi,\eta) = \Phi(\xi,\eta) + H_\Theta(\xi,\eta)$, we have

$$\lim_{\varepsilon\to 0}\int_{\partial B_\varepsilon(\eta)} u(\xi)\frac{\partial G_\Theta}{\partial \boldsymbol{v}_\varepsilon}(\xi,\eta)\,\mathrm{d}\sigma_\xi = u(\eta). \tag{18}$$

It remains to show that the second integral over $\partial B_\varepsilon(\eta)$ vanishes. Since $u \in C^2(\Omega)$, $\nabla u$ and $\frac{\partial u}{\partial \boldsymbol{v}_\varepsilon}$ is bounded in a neighborhood of $\eta$. Therefore,

$$\left|\frac{\partial u}{\partial \boldsymbol{v}_\varepsilon}\right| \le C_1. \tag{19}$$

Moreover,

$$G_\Theta(\xi,\eta) = \frac{1}{2\pi}\log\|\xi-\eta\| + H_\Theta(\xi,\eta), \tag{20}$$

where $H_\Theta$ is bounded near $\eta$. Therefore, on $\partial B_\varepsilon(\eta)$,

$$|G_\Theta(\xi,\eta)| \le C_2(1 + |\log\varepsilon|). \tag{21}$$

Consequently,

$$\left|\int_{\partial B_\varepsilon(\eta)} G_\Theta(\xi,\eta)\frac{\partial u}{\partial \boldsymbol{v}_\varepsilon}(\xi)\,\mathrm{d}\sigma_\xi\right| \le \int_{\partial B_\varepsilon(\eta)} |G_\Theta(\xi,\eta)|\left|\frac{\partial u}{\partial \boldsymbol{v}_\varepsilon}(\xi)\right|\,\mathrm{d}\sigma_\xi \le C_1C_2(1+|\log\varepsilon|)\cdot(2\pi\varepsilon) \to 0 \text{ as } \varepsilon\to 0. \tag{22}$$

Therefore, taking the limit $\varepsilon \to 0$ in Eq. (9), we obtain

$$u(\eta) = \int_{\partial\Omega} u(\xi)\frac{\partial G_\Theta}{\partial \mathbf{n}}(\xi,\eta)\,\mathrm{d}\sigma_\xi - \int_{\partial\Omega} G_\Theta(\xi,\eta)\frac{\partial u}{\partial \mathbf{n}}(\xi)\,\mathrm{d}\sigma_\xi. \tag{23}$$

and proves the identity. □

## B Further Insights on Theorem 3

In Theorem 3 of the main text, we prove that DVGC converges pointwise to GC when $R \to +\infty$ in the 2D case. Here, we provide more intuitive insights into this fact. We also discuss the high-dimensional case, as the 2D case is actually deserves special attention. Actually, when the dimension $d \ge 3$, the Green's function of the ball $B_R(\mathbf{0})$ converges to the

fundamental solution of the Laplace equation $\Phi(\xi,\eta)$ (i.e., $G_R(\xi,\eta) \longrightarrow \Phi(\xi,\eta)$ as $R \longrightarrow +\infty$). Consequently, DVGC $\longrightarrow$ GC will be easily vertified. For instance, the analytic Green's function for the three-dimensional ball is given by:

$$G_R(\xi,\eta) = -\frac{1}{4\pi\|\xi-\eta\|} + \frac{R}{4\pi\sqrt{|\xi|^2|\eta|^2 - 2R^2(\xi\cdot\eta) + R^4}}. \tag{24}$$

The gradient with respect to $\xi$ follows directly:

$$\nabla_\xi G_R(\xi,\eta) = \frac{1}{4\pi}\left(\frac{\xi-\eta}{\|\xi-\eta\|^3} - \frac{R(\|\eta\|^2\xi - R^2\eta)}{\left(\|\xi\|^2\|\eta\|^2 - 2R^2(\xi\cdot\eta) + R^4\right)^{3/2}}\right). \tag{25}$$

However, special treatment is required in the 2D case. This distinction arises because for $d \geq 3$, the fundamental solution of the Laplace equation converges ($\Phi(\xi,\eta) \longrightarrow 0$ as $\|\xi-\eta\| \longrightarrow +\infty$), whereas it does not converge when $d = 2$ ($\Phi(\xi,\eta) \longrightarrow +\infty$ as $\|\xi-\eta\| \longrightarrow +\infty$). However, the deformation results and the proof of Theorem 3 demonstrate that DVGC remains unaffected by the distinction in the $d = 2$ case, and its effect approximates that of the GC as $R$ increases. We will present a more macroscopic perspective below.

The first observation is that $G_R(\xi,\eta) + \frac{1}{2\pi}\log R \longrightarrow \Phi(\xi,\eta)$ as $R \longrightarrow +\infty$ in 2D. This can be easily verified as

$$\frac{R}{\sqrt[4]{R^4 - 2R^2(\xi\cdot\eta) + \|\xi\|^2\|\eta\|^2}} \longrightarrow 1, \qquad R \longrightarrow +\infty. \tag{26}$$

Therefore, if we define an equivalence class in the function space such that $f_1 \sim f_2$ if and only if there exists a constant $C$ such that $f_1 = f_2 + C$, we have $[G_R] \longrightarrow [\Phi]$ as $R \longrightarrow +\infty$. On the other hand, if $u$ is harmonic, the integral of the normal derivative of $u$ over the closed boundary equals zero. This can be easily proved using the divergence theorem as follows:

$$\int_{\partial\Omega} \frac{\partial u}{\partial \mathbf{n}}\, \mathrm{d}\sigma_\xi = \int_\Omega (\nabla\cdot\nabla u)\, \mathrm{d}\sigma_\xi = \int_\Omega \Delta u\, \mathrm{d}\sigma_\xi = 0. \tag{27}$$

Therefore, the following equation holds for arbitrary constant $C \in \mathbb{R}$:

$$u(\eta) = \int_{\partial\Omega} u(\xi)\frac{\partial G_\Theta}{\partial \mathbf{n}}(\xi,\eta)\, \mathrm{d}\sigma_\xi - \int_{\partial\Omega} (G_\Theta(\xi,\eta) + C)\frac{\partial u}{\partial \mathbf{n}}(\xi)\, \mathrm{d}\sigma_\xi \tag{28}$$

In other word, adding a (same) constant to all $\psi_j(\eta)$'s in GC and DVGC does not change the deformation effects. Given that $[G_R] \longrightarrow [\Phi]$ as $R \longrightarrow +\infty$, we conclude that GC is equivalent to DVGC as $R \longrightarrow +\infty$.

## C The Method of Images for Green's Functions on Disks

In the main text, we directly give the expressions of Green's functions on disks in Eqs. (3) and (4). We will provide a brief derivation here. Let $B_R(\mathbf{0}) \subset \mathbb{R}^2$ be the open disk centered at the origin with radius $R$. For a fixed $\eta \in B_R(\mathbf{0})$ and $\eta \neq \mathbf{0}$, we define its reflected point across $\partial B_R(\mathbf{0})$ as

$$\eta^* = \frac{R^2}{\|\eta\|^2}\eta. \tag{29}$$

We denote the Green's function for $B_R(\mathbf{0})$ as $G_R(\xi,\eta)$. Note that when $\|\xi\| = R$, since

$$\frac{\|\eta\|}{R}\|\xi - \eta^*\| = \sqrt{\frac{\|\xi\|^2\|\eta\|^2}{R^2} - 2(\xi\cdot\eta) + R^2}, \tag{30}$$

we have

$$\|\xi - \eta\| = \frac{\|\eta\|}{R}\|\xi - \eta^*\|. \tag{31}$$

Therefore, $G_R(\xi,\eta)$ can be derived explicitly by the method of images as:

$$G_R(\xi,\eta) = \Phi(\xi-\eta) - \Phi\left(\frac{\|\eta\|}{R}(\xi-\eta^*)\right), \tag{32}$$

which ensures that $G_R(\xi,\eta)=0$ for $\xi \in \partial B_R(\mathbf{0})$. Recall Eq. (30), we obtain the analytic solution of the Green's functions on disks as:

$$G_R(\xi,\eta) = \frac{1}{2\pi}\log\|\xi-\eta\| - \frac{1}{4\pi}\log\left(\|\xi\|^2\|\eta\|^2 - 2R^2(\xi\cdot\eta)+R^4\right) + \frac{1}{2\pi}\log R. \tag{33}$$

## D Distortion Analysis and Comparsions with Mean Value and Biharmonic Coordinates

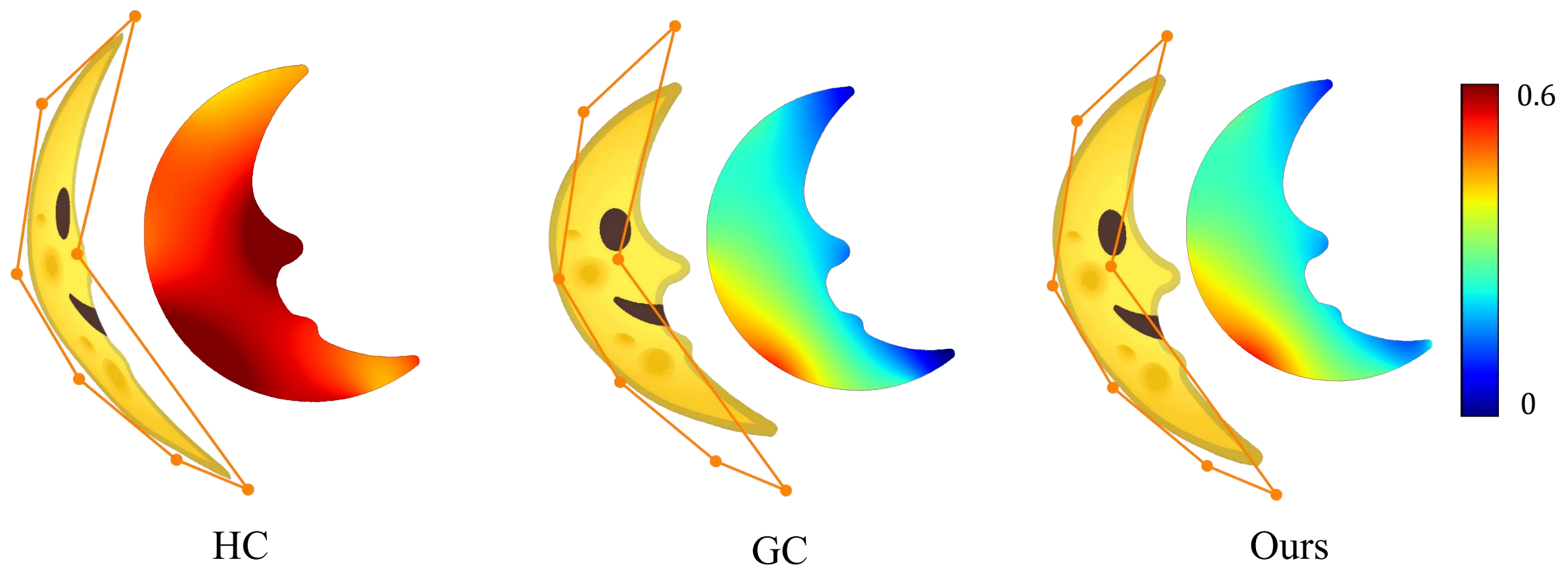


Fig. 2. Isometric distortion of our method, Harmonic Coordinates (HC) [Joshi et al. 2007] and Green Coordinates (GC) [Lipman et al. 2008].

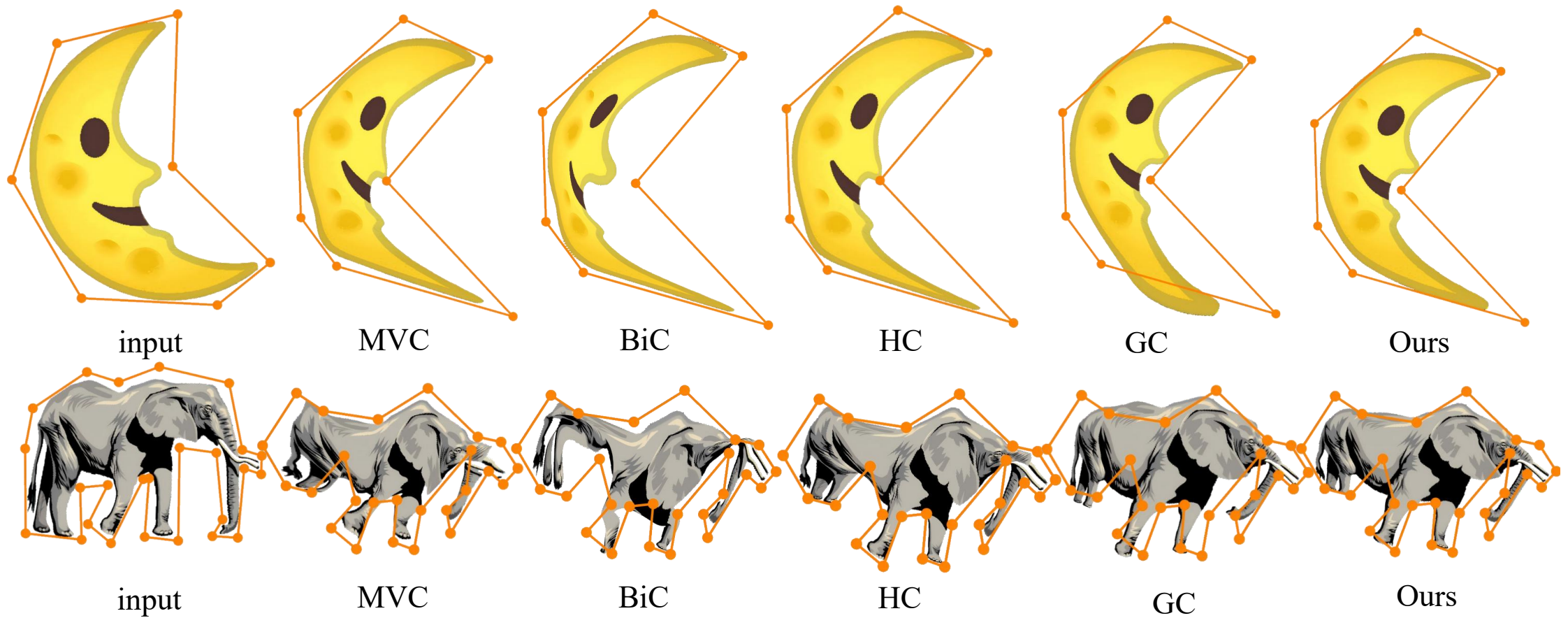


Fig. 3. Comparisons of our method with Mean Value Coordinates (MVC) [Floater 2003], Biharmonic Coordinates (BiC) [Weber et al. 2012], Harmonic Coordinates (HC) [Joshi et al. 2007] and Green Coordinates (GC) [Lipman et al. 2008].

This subsection presents additional comparisons to demonstrate the advantages of our method. Fig. 2 illustrates the isometric distortion of our method, Harmonic Coordinates (HC) [Joshi et al. 2007], and Green Coordinates (GC) [Lipman et al. 2008]. As observed, although HC yields deformations closer to the target cage, it introduces significant distortion. Both our method and GC exhibit the shape-preserving ability, resulting in relatively lower distortion. However, our method achieves a deformation that is notably closer to the target cage compared to GC.

In Fig. 3, we conduct comparisons with Mean Value Coordinates (MVC) [Floater 2003], Biharmonic Coordinates (BiC) [Weber et al. 2012], Harmonic Coordinates (HC) [Joshi et al. 2007] and Green Coordinates (GC) [Lipman et al. 2008]. We employ a minimal disk as the Green's function domain for the moon example, and a minimal rectangle for the elephant example of our method. We implement BiC and present the most basic deformation results introduced in Section 4.3 and Section 5.1 of [Weber et al. 2012]. Specifically, we construct the matrix composed of BiC, solve the system of equations and take the target face normal as the normal derivatives. It can be seen that although MVC, BiC, and HC remain close to the target cage, their deformations appear unnatural when the target cage undergoes significant shearing. In contrast, our method maintains shape-preserving while staying more consistent with the target cage than GC.


\end{document}